\documentclass[11pt]{article}
\usepackage{amsmath,amsfonts,amssymb,amsthm}
\usepackage{graphicx}
\usepackage{tikz}
\usepackage{enumerate}
\usepackage[hyphens]{url}
\usepackage{xcolor}
\usepackage{url}
\usepackage{caption}
\usepackage{subcaption}
\usepackage{tcolorbox}
\usepackage{hyperref}
\usepackage{chemarr}
\usepackage[margin=3cm]{geometry}
\usepackage{capt-of}
\usepackage{authblk}
\usepackage{booktabs}
\usepackage{placeins}
\usepackage{comment}
\usetikzlibrary{arrows.meta, automata, positioning} 

\theoremstyle{plain}
\newtheorem{theorem}{Theorem}

\theoremstyle{definition}
\newtheorem{definition}[theorem]{Definition}

\definecolor{darkgreen}{rgb}{0.0,0.7,0.0}

\usepackage{tikz}
\usetikzlibrary{patterns}
\usetikzlibrary{matrix, arrows, decorations.pathmorphing}
\usepackage{tikz-cd}
\tikzset{commutative diagrams/.cd}
\usepackage{framed,lipsum}
\usepackage{float}
\usepackage{pgfplots}
\usetikzlibrary{patterns}
\usetikzlibrary{intersections}
\usetikzlibrary{positioning, shadows, arrows}
\tikzstyle{every node}=[anchor=west, minimum height=3em]
\usetikzlibrary{decorations.pathmorphing}
\usepackage{enumitem}

\begin{document}


\title{ Mean-Field and Pairwise Approaches for the SIRI Model on Poisson Networks}

\author[1]{Akshara Bhat}
\author[1]{Abhishek Deshpande}
\author[1,2]{Chittaranjan Hens}
\author[3]{Subrata Ghosh}
\affil[1]{\small Center for Computational Natural Sciences and Bioinformatics, International Institute of Information Technology, Hyderabad} 
\affil[2]{\small Biomedical research center, International Institute of Information Technology, Hyderabad} 
\affil[3]{Division of Dynamics, Lodz University of Technology, Stefanowskiego 1/15, Lodz 90-924, Lodz, Poland}

\maketitle

\begin{abstract}       
Compartmental epidemic models, grounded in mass-action kinetics, often assume homogeneous mixing. Although this neglects network structure, recent results show that for Poisson random graphs, the classical SIR model, especially the susceptible decay curve, matches the susceptible decay dynamics of its network counterpart. Motivated by this, we investigate whether the extended SIRI model with relapse from the recovered class admits a similar correspondence.
SIRI dynamics arise in sevaral scenarios like spread of diseases with reactivation and behavioral contagion with relapse.  We derive parameter relationships under which the pairwise SIRI model on a Poisson network closely follows the mass-action ODE trajectories. When transmission per contact is small relative to recovery, the susceptible and infectious trajectories of both systems align. This establishes conditions under which nonlinear SIRI dynamics on networks can be effectively approximated by tractable mean-field equations.

\end{abstract}

\textbf{Keywords: Reaction Networks, SIRI model, Pairwise model, Poisson Network}

\maketitle 
    
\section{Introduction}
The modeling of infectious diseases is fundamental for understanding and mitigating epidemics. Classical compartmental models \cite{Brauer2008}, such as the Susceptible-Infected-Recovered (SIR) model \cite{Kermack1927}, \cite{Hethcote1989} and its extensions \cite{collapsetransition}, \cite{ghosh2021optimal}, have long been employed to describe the transmission dynamics of infectious diseases (e.g., \cite{AlQadi2022}). These models classify individuals into distinct compartments, namely Susceptible, Infected and Recovered and use mean-field differential equations to describe the rates of transition between these states.

At a fundamental level, the mathematical modeling of population dynamics, including infectious diseases, often employs principles from chemical reaction kinetics, specifically the law of mass action. This framework posits that the rate of a transition (or ``reaction") between population states is proportional to the product of the concentrations (or densities) of the interacting entities   \cite{keeling2008modeling,eames2002modeling}. Treating the population as a well-mixed chemical reactor, mass action kinetics provides the mathematical basis for the mean-field differential equations used in compartmental epidemiology. For instance, a unimolecular transition, $A \xrightarrow{\lambda} B$, yields the rate $\frac{d[A]}{dt} = -\lambda [A]$ \cite{ross1916application,de1995does}. Crucially, the fundamental component of disease spread, the infection event, is modeled as a bimolecular reaction, $S + I \xrightarrow{\beta} 2I$, where the infection rate($\beta SI$) is proportional to the product of Susceptible ($S$) and Infected ($I$) individuals, \cite{kermack1927contribution, hens2019spatiotemporal,ji2023signal}. The collection of these processes forms a reaction network, and these resulting mass action models form the mathematical bedrock for classical compartmental models like the SIR model.

Despite their utility, traditional compartmental models assume homogenous mixing, i.e., each individual has an equal probability of interacting with every other individual. This assumption often fails to capture the complexity of real-world social networks, where interactions are structured and heterogeneous. To address this limitation, network-based approaches have been developed, providing a more accurate representation of population structure \cite{pastor2015epidemic}. Many models of the transmission dynamics of infectious diseases (e.g., \cite{khudabukhsh2022projecting}, \cite{jacobsen2018large}, \cite{kiss2017mathematics}, \cite{gross2006epidemic}, \cite{ball2019stochastic}, \cite{risau2009contact}) represent contacts as a random graph of
$N$ individuals (nodes) formed using the configuration model \cite{molloy1995critical}. Generally the node degrees are assumed to be independent and identically distributed, such as in the Newman–Strogatz–Watts (NSW) random graph construction~\cite{newman2001random}. This interdisciplinary approach
helps in considering the nuances of real-world interactions and devising more realistic and impactful public health interventions \cite{iannelli2017effective}, \cite{roy2024impact}.

It has recently been shown in~\cite{rempala2023equivalence} that in the special case of a Poisson degree distribution, the deterministic SIR dynamics on a configuration model network is equivalent (in terms of the susceptible decay curve) to the classical mass-action SIR model. Specifically, the ``epidemic curve" describing the rate of decline in susceptibles is identical for both models when the transmission, recovery, and mean degree parameters are appropriately matched. While the infection prevalence curves may differ for small mean degrees, they converge and become practically indistinguishable for sufficiently large mean degrees. This equivalence helps explain the robustness of the mass-action formulation in approximating epidemic dynamics, even in structured populations, and provides theoretical support for using homogeneous-mixing models as surrogates for certain network-based processes.

The SIRI model is an extension of the SIR model that accounts for relapse of an infection in a recovered individual \cite{tudor1990deterministic}, \cite{moreira1997global}. This happens when an infected individual goes into a recovered state and subsequently experiences a reactivation of the infection. This recurrence of infection is observed in diseases such as bovine tuberculosis and human herpes \cite{tudor1990deterministic}, \cite{blower2004modelling} and can be modeled by SIRI epidemiological systems. Moreover, SIRI also shows promise in modeling the dynamics of tobacco and alcohol use as a contagious disease with relapse, as proposed in \cite{castillo1997mathematical} and \cite{sanchez2007drinking} respectively. The mean-field differential equations for this model (assuming constant population and no demographic changes) are as follows (\cite{tudor1990deterministic}, \cite{vandenDriessche2007}),
\let\saveeqnno\theequation
\let\savefrac\frac
\def\dispfrac{\displaystyle\savefrac}
\begin{eqnarray}
\let\frac\dispfrac
\gdef\theequation{1}
\let\theHequation\theequation
\label{eq:siri_massaction}
\begin{array}{@{}l}\begin{array}{l}dS/dt\;=\;-\beta SI\\dI/dt\;=\;\beta SI\;-\;\gamma I\;+\;\alpha R\\dR/dt\;=\;\gamma I\;-\;\alpha R\end{array}\end{array}
\end{eqnarray}
\global\let\theequation\saveeqnno
\addtocounter{equation}{-1}\ignorespaces 



\begin{figure}[ht]
  \centering
  \begin{tikzpicture}[
      ->, 
      >=Stealth, 
      node distance=2.5cm, 
      thick,
      every state/.style={draw, circle, minimum size=1cm}
    ]
    \node[state] (S) {S};
    \node[state] (I) [right=of S] {I};
    \node[state] (R) [below=of I] {R};

    \draw (S) to[bend left] node[above] {$\beta SI$} (I);
    \draw (I) to[bend left] node[right] {$\gamma I$} (R);
    
    \draw (R) to[bend left] node[left] {$\alpha R$} (I); 
    
  \end{tikzpicture}
  \caption{Transfer diagram for the SIRI ODE model in equations~\eqref{eq:siri_massaction}.}
  \label{fig:siri_diagram}
\end{figure}
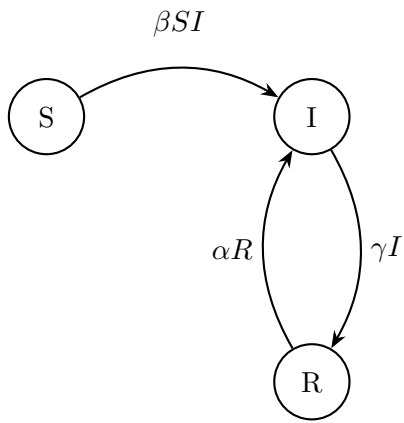

\begin{definition}[Model Parameters]
The constant parameters in the model are defined as follows:
\begin{itemize}
    \item $\beta$: The disease transmission coefficient, representing the rate of infection.
    \item $\gamma$: The recovery rate, representing the transition from infected to recovered.
    \item $\alpha$: The reinfection rate, indicating the speed at which immunity is lost.
\end{itemize}
\end{definition}
No birth and death rates are considered here.

While this model offers a more comprehensive framework for understanding disease dynamics, it still relies on the assumption of homogeneous mixing. Motivated by the proven equivalence between the classical SIR mass-action model and the SIR model on Poisson networks( to be introduced in the next section) \cite{rempala2023equivalence}, we investigate whether a similar correspondence can be established for the SIRI model, which includes relapse dynamics. Network-based SIRI models, such as the pairwise formulation, capture heterogeneous contact structures but are analytically and computationally more complex. By deriving parameter mappings between the pairwise SIRI equations on a Poisson network and the homogeneous-mixing SIRI ODEs, we aim to show that the susceptible and infectious trajectories of the network model can be closely approximated by the mass-action system. This approximation allows us to leverage the tractability of ODE-based methods while retaining high fidelity to the network-based dynamics, especially when the per-contact transmission rate is small relative to the recovery rate.

\subsection*{Network epidemic models}
We consider a stochastic SIRI (Susceptible{\textendash}Infectious{\textendash}Recovered{\textendash}Infectious) epidemic process evolving on a network of size $N$. Initially, $m$ individuals are selected uniformly at random from the population to be infectious. Each infectious individual remains infectious for a period drawn from an exponential distribution with rate \ensuremath{\tilde{\gamma} }. During this time, they contact their immediate neighbors. If a contacted neighbor is in the susceptible state at the moment of contact, they become infectious according to a poisson process with rate \ensuremath{\tilde{\beta} }. Once the infectious period ends, the individual enters the recovered state. Recovered individuals do not remain permanently immune; instead, they may relapse into the infectious state at a rate \ensuremath{\tilde{\alpha} }, with the time to relapse also modeled as an exponentially distributed random variable. All infection, recovery, and relapse processes are independent.

The epidemics unfolds on a random network generated via the configuration model. The network is built by drawing degree values for each node from a Poisson distribution and then joining the resulting half-edges, or stubs, at random to establish the edges. This results in a random graph with a Poisson degree distribution, capturing the contact structure of real populations while maintaining analytical tractability. In the following section, we describe the pairwise model, a mean-field model of the epidemics on a network.

\section{Pairwise model}
Approximating stochastic epidemics on networks is an important problem that has received significant attention, leading to the development of several mean-field models that are represented in terms of systems of ordinary differential equations. Notable among these models are the pairwise model \cite{rand1999correlation}, the Volz model \cite{volz2008sir}, and the dynamical survival analysis (DSA) model\cite{jacobsen2018large}, \cite{khudabukhsh2020survival}.
The pairwise model is based on a set of equations for the expected
number of susceptible ([$S$]) and infected ([$I$]) nodes, as well as the expected number of S–I
([$SI$]) and S–S ([$SS$]) pairs. It relies on a closure argument that approximates the expected number
of triples in terms of singles and pairs, which breaks the dependence on higher
moments. The pairwise model equations proposed by {\cite{rand1999correlation,keeling1997correlation}} are given by
\setcounter{equation}{1}
\begin{equation}\label{eq:pairwise_original_rand}
\begin{aligned}
\dot{[S]} &= -\tilde{\beta} [SI], \\
\dot{[I]} &= \tilde{\beta} [SI] - \tilde{\gamma} [I], \\
\dot{[R]} &= \tilde{\gamma} [I],\\
\dot{[SI]} &= -\tilde{\gamma} [SI] + \tilde{\beta}([SSI] - [ISI]) - \tilde{\beta} [SI], \\
\dot{[SS]} &= -2\tilde{\beta} [SSI], \\ 
\end{aligned}
\end{equation}
In a similar way, the pairwise model for SIRI can be captured by:
\begin{equation}\label{eq:pairwise_original}
\begin{aligned}
\dot{[S]} &= -\tilde{\beta} [SI],\\
\dot{[I]} &= \tilde{\beta} [SI] - \tilde{\gamma} [I] + \tilde{\alpha}[R],\\
\dot{[R]} &= \tilde{\gamma} [I] - \tilde{\alpha}[R],\\
\dot{[SI]} &= -\tilde{\gamma} [SI] + \tilde{\beta}([SSI] - [ISI]) - \tilde{\beta} [SI] + \tilde{\alpha}[SR],\\
\dot{[SS]} &= -2\tilde{\beta} [SSI],\\
\dot{[SR]} &= \tilde{\gamma}[SI] - \tilde{\alpha}[SR] - \tilde{\beta}[ISR].
\end{aligned}
\end{equation}
we define variables of the form 
\([X], [XY], [XYZ]\) below, 
\begin{definition}[Pairwise Variables]
Let the population be represented by an undirected graph $G$. We denote the expected counts of nodes and edges in specific states using the following notation:
\begin{itemize}
    \item $[X]$ denotes the expected number of nodes in state $X \in \{S, I, R\}$.
    \item $[XY]$ denotes the expected number of ordered edges $(u, v)$ such that node $u$ is in state $X$ and node $v$ is in state $Y$. Note that $[XX]$ is twice the number of undirected edges connecting two nodes in state $X$.
    \item $[XYZ]$ denotes the expected number of ordered triples of nodes $(u, v, w)$ connected as $u-v-w$ in states $X, Y$, and $Z$ respectively.
\end{itemize}
\end{definition} 
These are defined over the entire network such that each configuration is counted in all its possible permutations.
Here, we account for reinfection with an additional parameter $\tilde{\alpha}$, 
which is the probability with which a recovered individual experiences a relapse of the infection. This introduces a new term for the number of susceptible-recovered edges in $\dot{[SI]}$ dynamics. We gain [$SR$] pairs from [$SI$] pairs at the rate $\tilde{\gamma}$ as individuals recover and lose [$SR$] pairs as recovered individuals are reinfected or susceptible individuals are infected.

These equations do not include differential equations for triples that are involved in the calculation of pairs. The equations for triples will involve quadruples, and quadruples will involve pentuples, and so on, until the maximum possible degree. However, we apply moment closure on triples to express them in terms of pairs and ``close" the model.\\
\subsection{Poisson Networks}

We consider the special case of Poisson networks for the pairwise model. Triplets of the form $[XYZ]$ in Poisson type networks can be closed as, $[XYZ] \simeq \kappa\frac{[XY][YZ]}{[Y]}$ where $\kappa$ is the ratio of mean excess degree to mean degree of the network \cite{kiss2022necessary}. Given a random graph with degree distribution $p_k$ and finite average degree $\mu$ the excess degree distribution 
$q_k$ 
is defined as the degree of a random neighbor node of a randomly selected node of degree at least one. That is, 
$q_k = (k + 1)p_{k+1} /\mu$. 
On a Poisson network, 
$\kappa=1$ 
because the probability mass function for the Poisson distribution is stable under the transformation to the excess degree distribution.\\ 
\subsection{Normalized variables and useful ratios}
We introduce the normalized variables $x_S=\frac{[S]}{[N]}, x_I=\frac{[I]}{N}, x_R=\frac{[R]}{N}$ and variables $x_{SI}=\frac{[SI]}{N}, x_{SR}=\frac{[SR]}{N}, x_{SS}=\frac{[SS]}{N} $.
\begin{definition}[Auxiliary Ratios]
To simplify the pairwise system, we define the infectious-neighbor density $x_D$ and the recovered-neighbor density $x_A$ as:
$$x_D = \frac{x_{SI}}{x_S}, \quad x_A = \frac{x_{SR}}{x_S}$$.
\end{definition}
Using the closure with $\kappa=1$ and the factorization of $x_{SS}$ in terms of $x_S$, we will show that one can reach eqns.\ \eqref{eq:xdot_s} to \eqref{eq:xdot_a} from eqn.\ \eqref{eq:pairwise_original}.
\subsection*{Closed system on a Poisson network}
\begin{align}
\dot{x}_S &= -\tilde{\beta}\,x_D\,x_S, \label{eq:xdot_s}\\
\dot{x}_D &= \tilde{\beta}\,\mu\,x_S\,x_D - (\tilde{\beta}+\tilde{\gamma})\,x_D + \tilde{\alpha}\,x_A, \label{eq:xdot_d}\\
\dot{x}_I &= \tilde{\beta}\,x_D\,x_S - \tilde{\gamma}\,x_I + \tilde{\alpha}\,x_R, \label{eq:xdot_i}\\
\dot{x}_R &= \tilde{\gamma}\,x_I - \tilde{\alpha}\,x_R, \label{eq:xdot_r}\\
\dot{x}_A &= \tilde{\gamma}\,x_D - \tilde{\alpha}\,x_A. \label{eq:xdot_a}
\end{align}

\paragraph{Detailed derivation.}
We derive eqns.\ \eqref{eq:xdot_s}-\eqref{eq:xdot_a} from eqn.\eqref{eq:pairwise_original} step by step. 
Throughout, write densities $x_Z=[Z]/N$ and $x_{YZ}=[YZ]/N$.
We also use the standard triplet closure on configuration-model networks,
\begin{equation}
x_{XYZ}\;=\;\kappa\,\frac{x_{XY}\,x_{YZ}}{x_Y},
\qquad 
\kappa \;=\; \frac{\text{mean excess degree}}{\text{mean degree}},
\label{eq:closure-kappa}
\end{equation}
and, for Poisson degree, $\kappa=1$. Finally, on a Poisson random network, the pair density $x_{SS}$ can be factorized in terms of $x_S$. More precisely, one can show that (see \cite{kiss2022necessary}, Sec.~4):
\begin{equation}
x_{SS}\;=\;\mu\,x_S^2,
\label{eq:poisson-SS}
\end{equation}

\smallskip
In what follows, we derive the equations \eqref{eq:xdot_s}, \eqref{eq:xdot_i}, \eqref{eq:xdot_r}, \eqref{eq:xdot_d}, \eqref{eq:xdot_a} in this order. \\
Dividing first to third equations of eqn.\ \eqref{eq:pairwise_original} by $N$:
\[
\dot x_S \;=\; -\tilde\beta\,x_{SI},\qquad
\dot x_I \;=\; \tilde\beta\,x_{SI}-\tilde\gamma\,x_I+\tilde\alpha\,x_R,\qquad
\dot x_R \;=\; \tilde\gamma\,x_I-\tilde\alpha\,x_R.
\]
Using $x_{SI}=x_D\,x_S$ immediately gives
\begin{equation}
\begin{split}
\dot x_S &= -\tilde\beta\,x_D\,x_S, \\
\dot x_I &= \tilde\beta\,x_Dx_S-\tilde\gamma\,x_I+\tilde\alpha\,x_R, \\
\dot x_R &= \tilde\gamma\,x_I-\tilde\alpha\,x_R.
\end{split}
\label{eq:xdotS-deriv}
\end{equation}

This establishes equation \eqref{eq:xdot_s}, \eqref{eq:xdot_i} and \eqref{eq:xdot_r}. We can obtain the right–hand sides of  eqns. \eqref{eq:xdot_d} and \eqref{eq:xdot_a} once we identify triplet closures for pairs.

\begin{itemize}
\item \textbf{Triplet closures needed for pairs.}
Applying eqn.\ \eqref{eq:closure-kappa} with $\kappa=1$ and using eqn.\eqref{eq:poisson-SS} and $x_{SI}=x_Dx_S$, $x_{SR}=x_Ax_S$ gives
\begin{equation}
\begin{aligned}
x_{SSI} &\approx \frac{x_{SS}\,x_{SI}}{x_S}
= \mu x_S^2\cdot\frac{x_Dx_S}{x_S}
= \mu\,x_S^2\,x_D, \\[6pt]
x_{ISI} &\approx \frac{x_{SI}^2}{x_S}
= \frac{(x_Dx_S)^2}{x_S}
= x_D^2\,x_S, \\[6pt]
x_{ISR} &\approx \frac{x_{SI}x_{SR}}{x_S}
= x_D\,x_A\,x_S.
\end{aligned}
\label{eq:all-triplets}
\end{equation}
\end{itemize}

We will now derive equations for \eqref{eq:xdot_d} (equation for $x_D$) and \eqref{eq:xdot_a} (equation for $x_A$).
\begin{itemize}
\item \textbf{Equation for $x_D$.}
Start from the equation for $\dot{[SI]}$ in eqn.\eqref{eq:pairwise_original}, divide by $N$, and substitute eqn.\ \eqref{eq:all-triplets}:
\[
\dot x_{SI} \;=\; -\tilde\gamma\,x_{SI}+\tilde\beta\,(x_{SSI}-x_{ISI})-\tilde\beta\,x_{SI}+\tilde\alpha\,x_{SR}
\;=\;-(\tilde\beta+\tilde\gamma)\,x_{SI}+\tilde\beta\big(\mu x_S^2x_D-x_D^2x_S\big)+\tilde\alpha\,x_{SR}.
\]
Because $x_D=x_{SI}/x_S$, the quotient rule yields
\[
\dot x_D \;=\; \frac{\dot x_{SI}}{x_S}-x_D\,\frac{\dot x_S}{x_S}.
\]
Insert $\dot x_{SI}$ from above and $\dot x_S$ from eqn.\ \eqref{eq:xdotS-deriv}:
\begin{align*}
\dot x_D 
&= \frac{-(\tilde\beta+\tilde\gamma)\,x_{SI}+\tilde\beta(\mu x_S^2x_D-x_D^2x_S)+\tilde\alpha\,x_{SR}}{x_S}
  \;-\;x_D\Big(\frac{-\tilde\beta x_D x_S}{x_S}\Big),\\
&= -(\tilde\beta+\tilde\gamma)\,x_D+\tilde\beta(\mu x_Sx_D - x_D^2)+\tilde\alpha\,x_A \;+\;\tilde\beta x_D^2, \\
&= \tilde\beta\,\mu\,x_S\,x_D \;-\;(\tilde\beta+\tilde\gamma)\,x_D \;+\;\tilde\alpha\,x_A.
\end{align*}
which is exactly eqn.\ \eqref{eq:xdot_d}.
 
\item \textbf{Equation for $x_A$ .}
From eqn.\ \eqref{eq:pairwise_original}, (divide by $N$ and use eqn.\ \eqref{eq:all-triplets}):
\[
\dot x_{SR} \;=\; \tilde\gamma\,x_{SI}-\tilde\alpha\,x_{SR}-\tilde\beta\,x_{ISR}
\;=\; \tilde\gamma\,x_Dx_S - \tilde\alpha\,x_Ax_S - \tilde\beta\,x_Dx_Ax_S.
\]
Using $x_A=x_{SR}/x_S$, quotient rule and equation \eqref{eq:pairwise_original}, we get
\[
\dot x_A
= \frac{\dot x_{SR}}{x_S}-x_A\,\frac{\dot x_S}{x_S}
= \Big(\tilde\gamma x_D - \tilde\alpha x_A - \tilde\beta x_Dx_A\Big) \;-\; x_A\Big(\frac{-\tilde\beta x_D x_S}{x_S}\Big)
= \tilde\gamma\,x_D - \tilde\alpha\,x_A,
\]
which is eqn.\ \eqref{eq:xdot_a}.



\item \textbf{Conservation check.}
Adding eqn.\ \eqref{eq:xdot_s}, eqn.\ \eqref{eq:xdot_i}, eqn.\ \eqref{eq:xdot_r} gives $\dot x_S+\dot x_I+\dot x_R=0$, so $x_S+x_I+x_R\equiv 1$ as required.
\end{itemize}


\subsection*{Initial conditions}
Let the initial infected fraction be $0 < \tilde{\rho} << 1$, with no recovered individuals at $t=0$. For a Poisson network with mean degree $\mu$, the initial conditions are, 
\begin{equation}\label{eq:poisson_initial_conditions}
\begin{aligned}
x_S(0) &= 1, \\
x_I(0) &= \tilde{\rho}, \\
x_R(0) &= 0,\\
x_{SS}(0) &= \mu\,x_S(0)^2 = \mu,\\
x_{SI}(0) &= \mu\,x_S(0)\,x_I(0) = \mu\,\tilde{\rho},\\
x_D(0) &= \frac{x_{SI}(0)}{x_S(0)} = \mu\,\tilde{\rho},\\
x_A(0) &= \frac{x_{SR}(0)}{x_S(0)} = 0.
\end{aligned}
\end{equation}
These choices reflect random, uncorrelated seeding on a Poisson network: pairs factorize as products of node fractions, leading to $x_{SS}(0)=\mu x_S(0)^2$ and $x_{SI}(0)=\mu x_S(0) x_I(0)$.

\medskip
Equations \eqref{eq:xdot_s}-\eqref{eq:xdot_a} together with eqn.\ \eqref{eq:poisson_initial_conditions} form a closed, five-dimensional ODE model for SIRI dynamics on Poisson networks under pair approximation.

\section{Epidemic curve equation}
\begin{definition}[Epidemic Curve]
The \textit{epidemic curve} of a compartmental model is the functional representation of the rate of change of the susceptible population, $\dot{x}_S$, expressed as a fraction of $x_S$ and over time.
\end{definition}
Understanding the evolution of the susceptible population is central to epidemic modeling, and it provides a direct link between network-based models and their mass-action counterparts. One way to capture its evolution is through the \emph{epidemic curve equation}, which characterizes how the density of susceptibles decays over time as infections spread. This approach has been used\cite{rempala2023equivalence} in both classical ODE models and network-based pairwise models to compare their trajectories and derive conditions for equivalence. By analyzing the epidemic curve equations for both mass-action and pairwise systems, we can identify parameter mappings and structural assumptions under which the two formulations predict similar dynamics.\\
An important parameter in epidemiology is the basic reproduction number, $Z_0$ which indicates how many secondary infections are caused by an infectious individual. In the classical SIR model it is given by,
$Z^{MA} = \frac{\beta}{\gamma}$  
A similar parameter can be defined for the network SIR model \cite{rempala2023equivalence},
$Z^{NET} = \frac{{\mu} \tilde{\beta}}{\tilde{\beta} + \tilde{\gamma}}$. These constants are not the basic reproduction numbers for the SIRI model, rather, they only indicate the initial direction of infection curve. Nevertheless they will be useful in the upcoming calculations.\\
By dividing equation ~\eqref{eq:xdot_d} by ~\eqref{eq:xdot_s} and substituting $Z^{NET}$, we get:
\[
\frac{dx_D}{dx_S} = -{\mu} + \frac{{\mu}}{Z^{NET}x_S} - \frac{\tilde{\alpha} x_A}{\tilde{\beta} x_D x_S}.
\]
This differential equation is non-linear and cannot be solved analytically. However, we can still express $x_D$ as,\\
\[
x_D = \mu \int_{x_S(0)}^{x_S} \left( -1 + \frac{1}{Z^{NET}x_S} - \frac{\tilde{\alpha} x_A}{\tilde{\beta}\mu x_D x_S} \right)  dx_S + \mu\tilde{\rho}.
\]
Here the term $\mu \tilde{\rho}$ from the initial condition \eqref{eq:poisson_initial_conditions}, specifically $x_D = \mu \tilde{\rho}$.\\
Substituting into equation ~\eqref{eq:xdot_s}:
\begin{align}
    \dot{x}_S = -\tilde{\beta} {\mu} \left(\int_{x_S(0)}^{x_S} \left( -1 + \frac{1}{Z^{NET}x_S} - \frac{\tilde{\alpha} x_A}{\tilde{\beta}\mu x_D x_S} \right) dx_S + \tilde{\rho}\right) x_S.
    \label{eq:xdot_s_curve}
\end{align}
This equation is called the epidemic curve equation.
For the classical mass-action model, similar manipulation of eqns.\ (1-3) gives us,
\begin{align}
\dot{S} = -{\beta}\left(\int_{S(0)}^{S} \left( -1 + \frac{1}{Z^{MA}S} - \frac{{\alpha} R}{{\beta} I S} \right)  dS + \rho\right) S.
\label{eq:Sdot_curve}
\end{align}


\section{Model approximation}
In this section, we establish the conditions under which the pairwise SIRI model on a Poisson network can be faithfully approximated by the classical mass-action SIRI equations. The key idea is to compare the epidemic curve equations derived from both formulations and to identify parameter transformations that make the susceptible and infectious trajectories coincide. By carefully analyzing the relations between the densities of pairs and nodes, and by imposing suitable assumptions on the relative sizes of transmission and recovery rates, we show that the ratio of auxiliary variables in the pairwise system aligns with the corresponding ratios in the mass-action model. This allows us to derive a rigorous mapping between the two descriptions and to quantify when the simpler ODE formulation provides a valid surrogate for the network-based dynamics.

On comparing both the epidemic curve equations, for eqn.~\eqref{eq:siri_massaction} to approximate eqn. ~\eqref{eq:pairwise_original} we need the initial conditions and reinfection rate to be exactly the same and ${Z^{MA}}$ = $Z^{NET}$ which implies,
\begin{eqnarray}
\begin{split}
{\beta} &= \mu\tilde{\beta},\\
{\gamma} &= \tilde{\gamma},\\
 \alpha &= \tilde{\alpha},\\
\rho &= \tilde{\rho}.\\
\end{split}
\end{eqnarray}

\begin{theorem}[Model Approximation Validity]
The Susceptible curve described by \eqref{eq:siri_massaction} accurately approximates the network-based pairwise model in equation \eqref{eq:pairwise_original}throughout the epidemic if the following conditions hold:
\begin{enumerate}[label=\roman*)]
    \item $\beta = \mu \tilde{\beta}$,
    \item $\gamma = \tilde{\gamma}$,
    \item $\alpha = \tilde{\alpha}$,
    \item $\rho = \tilde{\rho}$,
    \item  $\tilde{\beta} << \tilde{\gamma}$.
\end{enumerate}

\end{theorem}

\textbf{Proof:}First we show that $ x_A/x_D \approx x_R/x_I $. From ~\eqref{eq:xdot_d} we can write,
\[
x_A = \frac{1}{\tilde{\alpha}}[\dot{x}_D + (\tilde{\beta} + \tilde{\gamma} - \tilde{\beta}{\mu} x_S )x_D].
\]
Dividing both sides by $x_D$ we obtain,\\
\begin{align} \label{eq:Ratio_A_D}
    \frac{x_A}{x_D} = \frac{1}{\tilde{\alpha}}[\frac{\dot{x}_D}{x_D} + (\tilde{\beta} + \tilde{\gamma} - \tilde{\beta}{\mu} x_S )].    
\end{align}
and similar manipulation of  Eqn.\ ~\eqref{eq:xdot_i} yields,\\
\begin{align} \label{eq:Ratio_R_I}
    \frac{x_R}{x_I} = \frac{1}{{\tilde{\alpha}}}[\frac{\dot{x_I}}{x_I} + (\tilde{\gamma} - \tilde{\beta} x_S x_D / x_I)].    
\end{align}
Comparing equations~\eqref{eq:Ratio_A_D} and ~\eqref{eq:Ratio_R_I}, we see that if, \\
\[
\frac{\dot{x_D}}{x_D} = \frac{\dot{x_I}}{x_I},
\qquad
\tilde{\beta} << \tilde{\gamma},
\qquad
x_D / x_I = \mu.
\]
the equations \eqref{eq:Ratio_A_D} and \eqref{eq:Ratio_R_I} match. We will show it for small enough \textit{t} using Taylor series expansion.

From Eqn.\eqref{eq:poisson_initial_conditions}, we can write the initial state of $x_D$ as, $x_D(0) = \mu\tilde{\rho} = \mu x_I(0)$. We can also write from eqn.\ ~\eqref{eq:xdot_d},\\
\[
\dot{x}_D(0) = \tilde{\beta}{\mu} x_D(0) - (\tilde{\beta} + \tilde{\gamma})x_D(0),
\]
\[
\dot{x}_D(0) = \tilde{\beta}{\mu^2}  x_I(0) - (\tilde{\beta} + \tilde{\gamma})\mu x_I(0),
\]
\[
\dot{x}_D(0) = \mu \dot{x}_I(0).
\]
For small enough \textit{t}, using Taylor series expansion we have,\\
\[
x_D(t) \approx x_D(0) + \dot{x}_D(0)t,
\]
Thus, we can write 
$
x_D(0) + \dot{x}_D(0)t = \mu x_I(0) + \mu \dot{x}_I(0) t.
$
As, 
 $x_I(0) + \dot{x}_I(0) t$ is the Taylor series approximation of $x_I(t)$, we can further write \\
$
x_D(0) + \dot{x}_D(0)t \approx  \mu {x}_I(t),
$
and approximately, one can write
$
x_D(t) \approx \mu {x_I}(t).
$
Thus,
\[
    \frac{\dot{x}_D}{x_D} = \frac{\dot{x}_I}{x_I}
    \label{eq:xd_dot/xd=xi_dot/xi},
    \qquad
    x_D / x_I = \mu,
\]
Substituting these results into the Eqns.\ \eqref{eq:Ratio_A_D} and ~\eqref{eq:Ratio_R_I} we obtain,\\
\begin{align}
    \frac{x_A}{x_D} = \frac{x_R}{x_I}.
\end{align}

Now it remains to show that $x_R/x_I = R/I$. Again since $\tilde{\beta} << \tilde{\gamma}$, we get $\gamma \approx \tilde{\gamma}$. Together with the conditions $\tilde{\beta} = \beta$ and $\tilde{\alpha} = \alpha$, we can write
$
\dot{x}_I = {\beta} x_{D}x_S - {\gamma} x_I + {\alpha} x_R,
$
and 
$
    \dot{x}_R = {{\gamma}} x_I - {\alpha} x_R.
$
Therefore, from Eqns.\ \eqref{eq:siri_massaction},  ~\eqref{eq:xdot_i}, and ~\eqref{eq:xdot_r}, we get $x_I = I, x_R = R $.
Substituting these into the Eqn.  ~\eqref{eq:xd_dot/xd=xi_dot/xi} we obtain,
$
\frac{x_A}{x_D} = \frac{R}{I}.
$
This leads us to,



\section{Evaluation and Discussion}

Figures~\ref{fig:full_comparison} shows the time series for both susceptible and infectious fractions for the four parameter sets. Each panel pairs network (solid) and ODE (dashed) curves. Despite relying on a first-order Taylor expansion in our theoretical derivation, the numerical results confirm that the theorem conditions lead to excellent agreement between the mass-action ODE and the network-based dynamics. 

\FloatBarrier 

\begin{figure}[htbp] 
    \centering
    \begin{minipage}[t]{0.49\linewidth} 
        \centering
        \includegraphics[width=\linewidth, keepaspectratio]{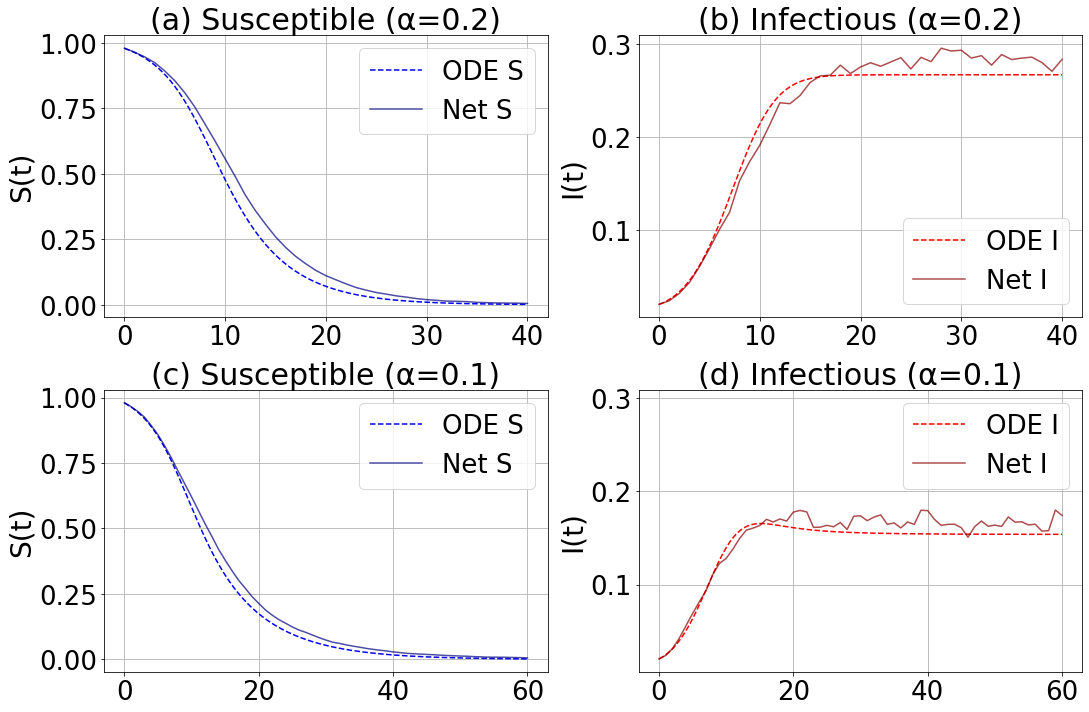}
    \end{minipage}
    \hfill
    \begin{minipage}[t]{0.49\linewidth} 
        \centering
        \includegraphics[width=\linewidth, keepaspectratio]{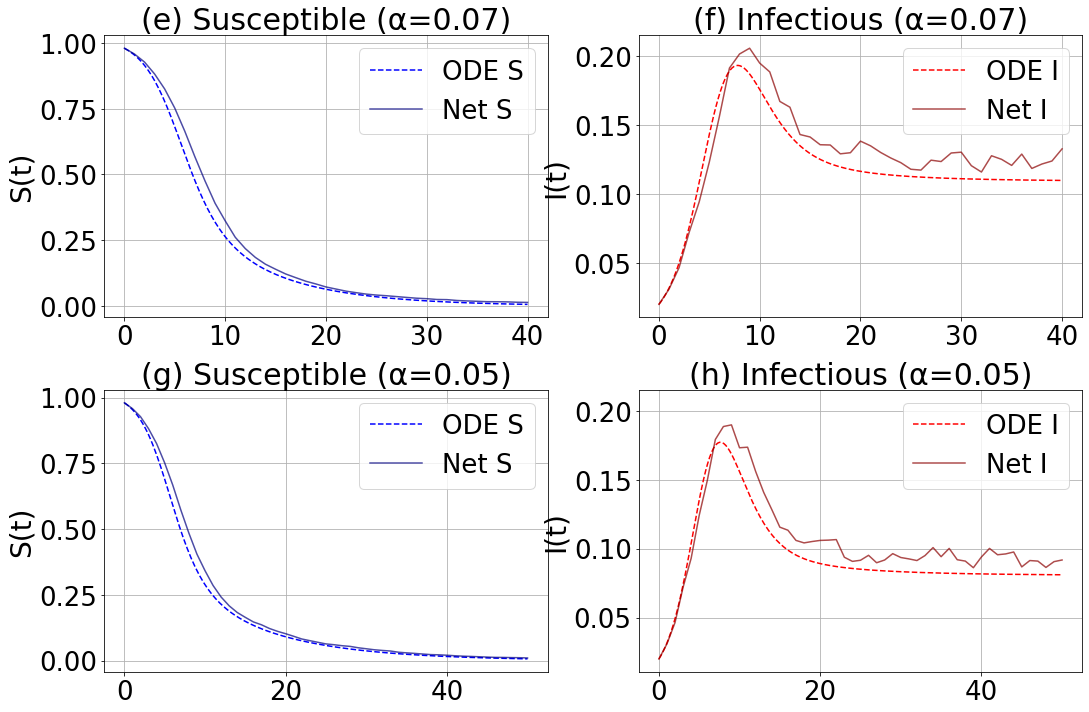}
    \end{minipage}

    \caption{Temporal evolution of the epidemic dynamics in a Poisson network (solid) compared to the corresponding mass-action model ODEs (dashed). The plots are arranged in a 2$\times$2 grid (S and I fractions for two parameter sets) within each image. The common parameters for all panels are average degree $\mu=15$, recovery rate $\tilde{\gamma}=0.5$, $N=5000$, and initial infected fraction $\tilde{\rho}=0.02$. The specific parameters for the four sets are detailed in the sub-captions: plots (a), (b), (c) and (d) use $\tilde{\beta}=0.05$, while plots (e), (f), (g) and (h) use $\tilde{\beta}=0.07$. $\tilde{\alpha}$ varies as shown. We have quantitatively evaluated our approximation using RMSE and Pearson Correlation Coefficient and results for the same are shown below.}
    \label{fig:full_comparison}
\end{figure}
\FloatBarrier 

\begin{figure}[htbp]
    \centering
    \begin{subfigure}[t]{0.48\linewidth}
        \centering
        \includegraphics[width=\linewidth]{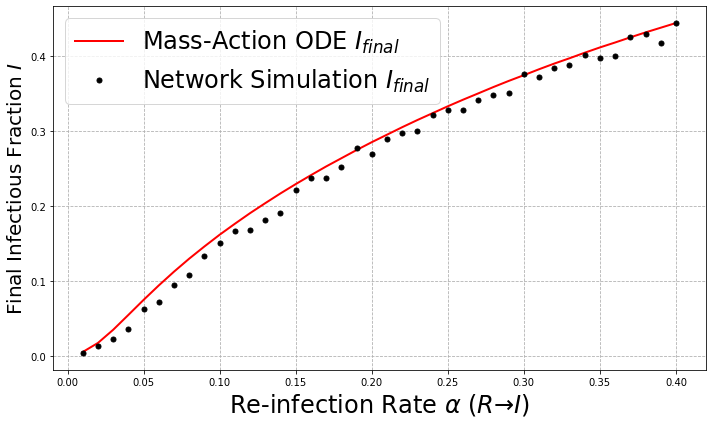}
        \caption{}
        \label{fig:final_fraction_alpha}
    \end{subfigure}
    \hfill
    \begin{subfigure}[t]{0.48\linewidth}
        \centering
        \includegraphics[width=\linewidth]{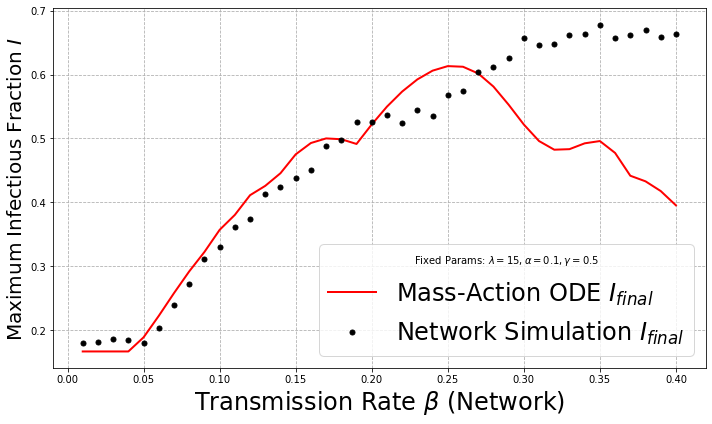}
        \caption{}
        \label{fig:final_fraction_beta}
    \end{subfigure}

    \caption{Steady-state analysis of the SIRI model: Comparison between the Mass-Action ODE (solid red line) and Network Simulation (black dots). Panel (a) shows the final infectious fraction as a function of the re-infection rate $\alpha$ with $\beta=0.05$. Panel (b) illustrates the infectious fraction relative to the network transmission rate $\beta$ with $\alpha=0.1$. Fixed parameters: $\lambda=6, , \gamma=0.5$.}
    \label{fig:final_analysis}
\end{figure}

    
\subsection*{Results Summary}
Across all four scenarios in \ref{fig:full_comparison}, the epidemic curve of mass-action SIRI ODE closely tracks the epidemic curve with network-based dynamics: S‑RMSE stays below 0.041 and correlation coefficients exceed 0.96 in all cases. We also observe that mass-action SIRI ODE accurately models the infectious fraction of the epidemic in a poisson network. This strong quantitative agreement, despite our use of a Taylor series closure in deriving the theorem, confirms that the mass-action approximation is valid under the stated conditions. 

Additionally we can see in Figure \ref{fig:final_analysis} and Figure \ref{fig:steady_state}, that the approximation holds with high precision for the final infectious ($I_{final}$) and recovered ($R$) fractions across varying re-infection and transmission rates. Notably, Figure \ref{fig:final_fraction_beta} demonstrates that while the ODE aligns perfectly with the network simulation at low transmission rates ($\beta$), a visible deviation emerges as $\beta$ increases, confirming our prediction. Similarly, Figure \ref{fig:steady_state} shows that for low enough infection rates, the steady-state recovered fraction stabilizes in close agreement with the ODE model, effectively capturing the transition at the epidemic threshold.

These results demonstrate that for $\tilde{\beta} \ll \tilde{\gamma}$ and matched initial conditions, the simpler ODE captures both the susceptible and infectious trajectories on Poisson networks with high fidelity. Table \ref{tab:results} below tabulates the results from our evaluation of RMSE and Pearson correlation coefficient on our approximation model.
\begin{figure}[!ht]
    \centering
    \includegraphics[width=0.5\linewidth]{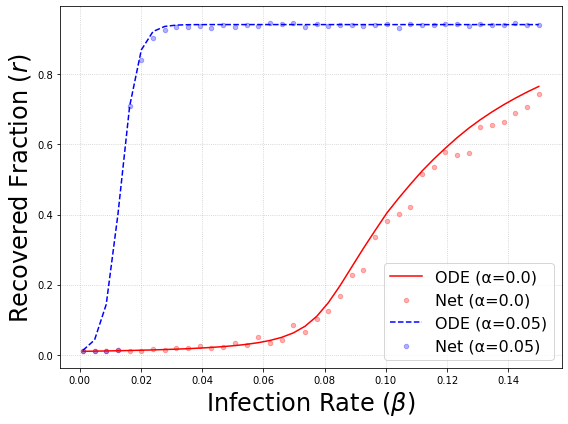}
    \caption{Steady-state Recovered Fraction ($r$) vs. Infection Rate ($\beta$) for different re-infection rates ($\alpha=0.0$ and $\alpha=0.05$). We take low enough infection rates for the approximation to work and steady state recovered fraction is measured when the epidemic stabilises.Fixed parameters: $\lambda=10,  \gamma=0.8$}
    \label{fig:steady_state}
\end{figure}

\begin{table}[H]
\centering
\caption{Quantitative comparison of network and ODE dynamics across four parameter sets. S-RMSE and I-RMSE denote the Root Mean Square Error for the susceptible and infectious fractions, respectively. Corr$_S$ and Corr$_I$ are the corresponding Pearson correlation coefficients.}
\label{tab:results}
\begin{tabular}{lcccccccc}
\toprule
\textbf{Run} & $\boldsymbol{\mu}$ & $\boldsymbol{\tilde{\beta}}$ & $\boldsymbol{\tilde{\alpha}}$ & $\boldsymbol{\tilde{\gamma}}$ & \textbf{S-RMSE} & \textbf{I-RMSE} & \textbf{Corr}$\boldsymbol{_S}$ & \textbf{Corr}$\boldsymbol{_I}$ \\
\midrule
1 & 15 & 0.05 & 0.20 & 0.5 & 0.0409 & 0.0146 & 0.9969 & 0.9902 \\
2 & 15 & 0.05 & 0.10 & 0.5 & 0.0259 & 0.0116 & 0.9986 & 0.9797 \\
3 & 15 & 0.07 & 0.07 & 0.5 & 0.0335 & 0.0141 & 0.9972 & 0.9657 \\
4 & 15 & 0.07 & 0.05 & 0.5 & 0.0265 & 0.0134 & 0.9980 & 0.9681 \\
\bottomrule
\end{tabular}
\end{table}

\subsection*{Impact of Re-infection rate($\tilde{\alpha}$)}
From the simulation results, we also see the effect of the reinfection rate $\tilde{\alpha}$. A separate simulation with the same parameters as Run-1 but with $\tilde{\alpha} = 0.5$ gave \textbf{S-RMSE} = 0.0636 and \textbf{I‑RMSE} = 0.0279. This shows us that a higher $\tilde{\alpha}$ leads to a poorer approximation and higher Root Mean Square Error. This is because the derivation of the epidemic curve equation, specifically the approximation of the network dynamics by the mass-action ODE, relies on a first-order Taylor series expansion and mean-field assumptions.
The approximation hinges on approximating the term $\frac{\tilde{\alpha} x_A}{\tilde{\beta}\mu x_D x_S}$ from \eqref{eq:xdot_s_curve} with $\frac{{\alpha} R}{{\beta} I S}$ from \eqref{eq:Sdot_curve} which involves the approximation $\frac{x_A}{x_D} \approx \frac{R}{I}$. 
The parameter $\tilde{\alpha}$ acts as a multiplicative factor for this approximated term. When $\tilde{\alpha}$ is large, it amplifies the error introduced by the approximation of the $\frac{x_A}{x_D}$ ratio, causing the network dynamics to diverge more significantly from the ODE solution. Essentially, a higher $\tilde{\alpha}$ weights the least accurate part of the mean-field approximation more heavily, thus resulting in the observed increase in the Root Mean Square Error (RMSE) and a poorer overall agreement between the network and ODE models.

\section{Conclusion}
In this work, we demonstrated that the classical mass-action SIRI model can effectively approximate the pairwise epidemic dynamics of a SIRI process on Poisson random networks under specific parameter mappings. By aligning key parameters—namely transmission, recovery, and reinfection rates—and ensuring equivalence in initial conditions, we derived a condition under which the trajectory of the susceptible population in the mass-action model closely mirrors that of the more detailed network-based pairwise model.

Our analytical derivation revealed that the approximation holds particularly well when the per-contact transmission rate is small relative to the recovery rate. Furthermore, we showed that the susceptible and infected curves predicted by the mass-action model can serve as reliable proxies for the corresponding network dynamics when moment closure assumptions are valid and the network structure is sufficiently random.

These findings have practical significance: they justify the use of simpler mass-action ODE models for forecasting and control in scenarios where full network information is unavailable or computationally expensive to simulate. Future work could investigate the robustness of this approximation under varying network topologies, clustering, or more complex disease dynamics such as multi-stage infections or adaptive behaviors.

\nocite{*}
\bibliographystyle{amsplain}
\bibliography{sn-bibiliography}
\end{document}